\documentclass[12pt,draftcls,onecolumn]{IEEEtran}
\usepackage{xcolor}
\usepackage{tikz}
\usepackage{graphicx}
\usepackage[hidelinks]{hyperref}
\usepackage{amsmath}
\usepackage{dsfont}
\usepackage{circuitikz}
\usepackage{tikz-cd}
\usepackage{enumerate}
\usetikzlibrary{arrows,positioning}
\tikzset{
  shift left/.style ={commutative diagrams/shift left={#1}},
  shift right/.style={commutative diagrams/shift right={#1}}
}

\usepackage{fix-cm}

\usepackage{algorithm}
\usepackage{algpseudocode}

\makeatletter
\def\BState{\State\hskip-\ALG@thistlm}
\makeatother

\usepackage{amsmath,amsthm,amssymb}
\usepackage{tikz}
\usepackage{graphicx}
\usepackage[hidelinks]{hyperref}
\usepackage{amsmath}
\usepackage{dsfont}
\usepackage{circuitikz}

\newtheorem{lem}{Lemma}

\newtheorem{rem}{Remark}
\newtheorem{defn}{Definition}
\newtheorem{cor}{Corollary}
\newtheorem{assum}{Assumption}
\newtheorem{examp}{Example}

\newtheorem{claim}{Claim}

\newtheorem{theo}{Theorem}

\newtheorem{pb}{Problem}
\newcommand{\tp}{\intercal}		

\DeclareMathOperator{\ee}{\mathbb{E}}			
    
\DeclareMathOperator\VVEC{\textsl{vec}}

\ifCLASSINFOpdf
\else
\fi
\begin{document}
%
\title{{\LARGE Dynamic Teams and Decentralized Control Problems with Substitutable Actions}}
%
%
%

\author{Seyed Mohammad Asghari and Ashutosh Nayyar \vspace{-6mm}
\thanks{Preliminary version of this paper appeared in the proceedings of the 54th annual Conference on Decision and Control (CDC), 2015 (see \cite{AsghariNayyar:2015}).}
\thanks{S. M. Asghari and A. Nayyar are with the Department of
Electrical Engineering, University of Southern California, Los Angeles,
CA 90089 USA (e-mail: asgharip@usc.edu; ashutosn@usc.edu).} 
 \thanks{This research was supported by NSF under grants ECCS 1509812 and CNS 1446901.}
}
\vspace{-20mm}
\maketitle
\begin{abstract}
This paper considers two problems -- a dynamic team problem and a decentralized control problem.
The problems we consider do not belong to the known classes of ``simpler'' dynamic team/decentralized control problems such as partially nested or quadratically invariant problems. However, we show that our problems admit simple solutions under an assumption referred to as the substitutability assumption.  
Intuitively, substitutability in a team (resp. decentralized control) problem means that the effects of one team member's (resp. controller's) action on the cost function and the information (resp. state dynamics) can be achieved by an action of another member (resp. controller).
For the non-partially-nested LQG dynamic team problem, it is shown that under certain conditions linear strategies are optimal.  For the non-partially-nested decentralized LQG control problem, the state structure can be exploited  to obtain optimal control strategies with recursively update-able sufficient statistics. 
These results suggest that  substitutability can work as a counterpart of  the information structure requirements that enable simplification of dynamic teams and decentralized control problems. 
%
\end{abstract}


%
\IEEEpeerreviewmaketitle

\vspace{-2mm}
\section{Introduction}\label{sec:intro}

%
%

The difficulty of finding optimal strategies in dynamic teams and decentralized control problems has been well-established in the literature \cite{HoChu:1972, witsenhausen, LipsaMartins:2011b,blondel}. In general, the optimization of strategies can be a non-convex problem over infinite dimensional spaces \cite{YukselBasar:2013}. Even the celebrated linear quadratic Gaussian (LQG) model of centralized control presents difficulties in the decentralized setting \cite{ witsenhausen, LipsaMartins:2011b, mahajan_survey}. There has been significant interest in identifying classes of problems that are more tractable. Information structures, which describe what information is available to which member/controller,  have been closely associated with their tractability. Problems with partially nested \cite{HoChu:1972} or stochastically nested information structures \cite{Yuksel:2009} and problems that satisfy quadratic invariance \cite{RotkowitzLall:2006} or funnel causality \cite{BamiehVoulgaris:2005} properties have been identified as  being ``simpler'' than the general problems. 




In this paper, we first look at the nature of cost function and the information dynamics in a LQG  dynamic team problem. We define a property called substitutability which 
means that the effects of one member's action on the cost function and the information dynamics can be achieved by the action of another member. Although the problem we formulate is not partially nested, our result shows that, under certain conditions, linear strategies are optimal.
We then consider a decentralized LQG problem and show that the idea of substitutability can be used in such problems as well. Substitutability in a decentralized problem can be interpreted as follows: the effects of one controller's action on the instantaneous cost and the state dynamics can be achieved by the action of another controller. Even though the problem we formulate does not belong to one of the simpler classes mentioned earlier, our results show that linear strategies are optimal. Further,  we provide a complete state-space characterization of optimal strategies and  identify a family of information structures that all achieve the same cost as the centralized information structure. These results suggest that substitutability can work as a counterpart of the information structure requirements that enable simplification of dynamic teams and decentralized control problems.  

Our work shares conceptual similarities with the work on internal quadratic variance \cite{lessard_iqi, Lessard:phd} which identified problems that are not quadratically invariant but can still be reduced to (infinite dimensional) convex programs. In contrast to this work, we explicitly identify optimal  strategies in the decentralized control problem. The interplay of information structure and cost in relation to  the complexity of dynamic team problems has also been observed for variations of the Witsenhausen counterexample \cite{witsenhausen} in  \cite{bansal, rotkowitz_cdc_06}.


\vspace{-5mm}
\subsection{Notation}
Uppercase letters denote random variables/vectors and their corresponding realizations are represented by lowercase letters. Uppercase letters are also used to denote matrices. 
For two functions $f$ and $g$ and a random variable/vector $X$, $f(X)=g(X)$ is interpreted as follows: for every realization $x$ of $X$, the realizations of $f(X)$ and $g(X)$ are equal to each other. 
$\ee[\cdot]$ denotes the expectation of a random variable. For a collection of functions $\boldsymbol{g}$, 
$\ee^{\boldsymbol{g}}[\cdot]$ denotes that the expectation depends on the choice of functions in $\boldsymbol{g}$.
When random vector $X$ is normally distributed with mean $\mu$ and variance $\Sigma$, it is shown as $X \sim \mathcal{N}(\mu, \Sigma)$.

For a sequence of column vectors $X, Y, Z,...$, the notation $\textsl{vec}(X,Y,Z,...)$ denotes vector $[X^{\tp}, Y^{\tp}, Z^{\tp},...]^{\tp}$. The vector $\textsl{vec}(X_1, X_2,...,X_t)$ is denoted by $X_{1:t}$. 
In addition, for a sequence of column vectors $X^i$, $i \in \mathcal{A} = \{\alpha_1, \alpha_2, \ldots, \alpha_n\}$, $X^{\mathcal{A}}$ is used to denote the vector $\textsl{vec}(X^{\alpha_1}, X^{\alpha_2}, \ldots,X^{\alpha_n})$.
The transpose and Moore-Penrose pseudo-inverse of matrix $A$ are denoted by $A^{\tp}$ and $A^{\dagger}$, respectively. For two vectors $X$ and $Y$, we use $X \subset_v Y$ to denote that $X$ is a sub-vector of $Y$ and $X \not \subset_v Y$ to denote that $X$ is not a  sub-vector of $Y$. 
If $\mathcal{A}$ is a set, we denote the cardinality of $\mathcal{A}$ by $\left\vert{\mathcal{A}}\right\vert$. 
Furthermore, if $\mathcal{A} = \{\alpha_1, \alpha_2, \ldots, \alpha_n\}$, we use $\{B^{m}\}_{m \in \mathcal{A}}$ to denote a matrix composed of $B^{\alpha_1}, B^{\alpha_2}, \ldots, B^{\alpha_n}$ as row blocks, that is 
$\{B^{m}\}_{m \in \mathcal{A}} = [(B^{\alpha_1})^{\tp}, (B^{\alpha_2})^{\tp},\ldots, (B^{\alpha_n})^{\tp}]^{\tp}$. 
Similarly, $\{B^{mk}\}_{m \in \mathcal{A}}$ denotes 
a matrix composed of $B^{\alpha_1 k}, B^{\alpha_2 k}, \ldots, B^{\alpha_n k}$ as row blocks.








\vspace{-2mm}
\section{Dynamic Team with Non-partially-nested Information Structure}
\label{dynamic_team}
\subsection{Team Model and Information Structure}
\label{Team_mo}
We consider a team composed of $n$ members. The set $\mathcal{M} = \{1, 2, ..., n \}$ denotes the collection of team members. 
The random vector $\Xi$  taking values in $\mathbb{R}^{d_{\xi}}$ denotes all the exogenous uncertainties  which are not controlled by any of the members. The probability distribution of $\Xi$ is assumed to be  $\mathcal{N}(0,\Sigma)$, where $\Sigma$ is a positive definite matrix.

The information available to member $i$ is denoted by $Z^i \in \mathbb{R}^{d^i_z}$. 
Member $i$ chooses  action/decision 
$U^i \in \mathbb{R}^{d^i_u}$ as a function of the information available to it. Specifically, for $i \in \mathcal{M}$, $U^i = g^i (Z^i)$ where $g^i$ is the decision strategy of member $i$. The collection $\mathbf{g} = (g^1, g^2, \ldots, g^n)$ is called the team strategy. The performance of the team strategy $\mathbf{g}$ is measured by the expected cost 
\begin{align}
\label{cost_team}
&\mathcal{J}(\mathbf{g}) = \ee^{\mathbf{g}}\left[
       (M\Xi + NU)^{\tp}(M\Xi + NU)\right]
\end{align}
where $U = \VVEC(U^1,\ldots,U^n)$ and $ N = \begin{bmatrix} N^1 &\ldots  &N^n\end{bmatrix}$.

The information $Z^i$ available to member $i$ includes what it has observed and what other members have communicated to it. We assume that $Z^i$ is a known linear function of $\Xi$ and the decisions taken by some other members, that is,

\vspace{-5mm}
\begin{align}
\label{inf.struct.general}
Z^i = H^i \Xi + \sum_{j \in \mathcal{M} \setminus \{i\}} D^{ij} U^j  \hspace{9mm} \forall i \in \mathcal{M}
\end{align}
where $H^i$ and $D^{ij}$ are matrices with appropriate dimensions. We assume that members make decisions sequentially according to their index and that the information of member $i$ can depend only on the decisions of members indexed $1$ to $i-1$. Thus, we assume that 
\begin{align}
\label{fact_team}
D^{ij} = 0  \hspace{2mm} \mbox{for~}  j \geq i.
\end{align}

The matrices $H^i,D^{ij}, i,j \in \mathcal{M},$ in the information structure, the matrices $M, N^i, i \in \mathcal{M},$ in the cost function, and the probability distribution of random vector  $\Xi$ are known to all team members. 


Following \cite{HoChu:1972}, we define the following relationships among team members.
\begin{defn}
We say that member $s$ is related to  member $t$ and denote this  by $sRt$ if $D^{ts} \neq 0$.
Further, we say  that member $s$ is a precedent of  member $t$ and denote this by $s \rightsquigarrow t$ if (a) $sRt$, or (b) there exist distinct $k_1, k_2, \ldots, k_m \in \mathcal{M}$ such that $sRk_1$, $k_1Rk_2,\ldots, k_mRt$.
\end{defn}
We denote the set of all precedents of member $t$ by $\mathcal{P}^t$. The team is said to have a \emph{partially nested information structure} if for each member $t$ and each $s \in \mathcal{P}^t$, $Z^s \subset_v Z^t$. In other words, whenever  the decision of member $s$ affects the information of member $t$, then $t$ knows whatever $s$ knows. For partially nested information structure, optimal strategies can be obtained using the method described in \cite{HoChu:1972}. We will focus on teams where the information structure is not partially nested. 

\begin{defn}
\label{critical_pair_defn}
We say $(s,t)$ is a critical pair with respect to partial nestedness  if $s \rightsquigarrow t $ but $Z^s \not \subset_v Z^t$.
\end{defn}
We denote the set of all members $s \in \mathcal{P}^t$ for which $(s,t)$ is a critical pair by $\mathcal{C}^t$. According to above definitions, an information structure is not partially nested if there exists $t \in \mathcal{M}$ for which $\mathcal{C}^t \neq \varnothing$.
\vspace{-3mm}
\subsection{Substitutability Assumption}
We make the following assumption about the team model.
\begin{assum}
\label{asum_subs_team}
For every critical pair $(s,t)$, there exists a member $k$ such that
\begin{enumerate}
\item $Z^s \subset_v Z^k$, and
\item  For every $u^t \in \mathbb{R}^{d^t_u}$, there exists a $u^k \in \mathbb{R}^{d^k_u}$ such that 
\begin{subequations}
\begin{align}
\label{subs_condition_1}
N^t u^t &= N^k u^k, \\
\label{subs_condition_2}
D^{mt} u^t &= D^{mk} u^k  \hspace{7mm} \forall m \in \mathcal{M}.
\end{align}
\end{subequations}
We  refer to member $k$ as the Substituting Member for the critical pair $(s,t)$.
\end{enumerate}
\end{assum}

\begin{examp}
Consider Problem \ref{Prob_1} with the following information structure:
\begin{align}
\label{example_info}
Z^1 &= H^1 \Xi, \hspace{3mm} Z^2 = H^2 \Xi + D^{21}U^1 ~~ (D^{21} \neq 0, Z^1 \not \subset_v Z^2)\nonumber \\
Z^3 & = H^3 \Xi + D^{31}U^1
=\begin{bmatrix}
H^1 \\
H^{3'}
\end{bmatrix} \Xi + \begin{bmatrix}
0 \\
D^{31'}
\end{bmatrix} U^1    \nonumber \\
Z^4 & = H^4 \Xi + D^{41}U^1 + D^{42}U^2 + D^{43}U^3 \nonumber \\
&=\begin{bmatrix}
H^1 \\
H^{3'} \\
H^2 \\
H^{4'}
\end{bmatrix} \Xi + \begin{bmatrix}
0 \\
D^{31'} \\
D^{21} \\
0
\end{bmatrix} U^1 + \begin{bmatrix}
0 \\
0 \\
0 \\
D^{42'}
\end{bmatrix} U^2
+ \begin{bmatrix}
0 \\
0 \\
0 \\
D^{43'}
\end{bmatrix} U^3 .
\end{align}
In this information structure, $(1,2)$ is a critical pair since $D^{21} \neq 0$ but $Z^1 \not \subset_v Z^2$. Since $Z_1 \subset_v Z_3$, the substitutability assumption may be satisfied in this example if for every $u^2$, there exists a $u^3$ satisfying $N^2 u^2 = N^3 u^3$ and $D^{42}u^2 = D^{43} u^3$. 
\end{examp}

\vspace{-2mm}

\begin{rem}
The substituting member $k$ for the critical pair $(s,t)$ can be the member $s$ itself as long as condition 2 of Assumption \ref{asum_subs_team} can be satisfied.
\end{rem}

\begin{rem}
In order to check condition 2 of Assumption \ref{asum_subs_team}, we need to verify that the column space of the matrix $\begin{bmatrix}
N^t \\
\{D^{mt}\}_{m \in \mathcal{M}}
\end{bmatrix}$ is contained in the column space of the matrix $\begin{bmatrix}
N^k \\
\{D^{mk}\}_{m \in \mathcal{M}}
\end{bmatrix}$. This can be easily done, for instance, by projecting the columns of the first matrix onto the column space of the second and verifying that the projection leaves the columns unchanged. Alternatively, for each column $c$ of the first matrix, we can check if the following equation has a solution for $x$: $\begin{bmatrix}
N^k \\
\{D^{mk}\}_{m \in \mathcal{M}}
\end{bmatrix} x=c$.
\end{rem}

The following lemma is immediate from the theory of pseudo-inverses \cite{ben2006generalized}.
\begin{lem}
\label{lem_team}
If a solution $u^k$ to \eqref{subs_condition_1} and \eqref{subs_condition_2} exists, it can be written as
\begin{align*}
u^k =  \Lambda^{kst}u^t = \begin{bmatrix}
N^k \\
\{D^{mk}\}_{m \in \mathcal{M}}
\end{bmatrix}^{\dagger}
\begin{bmatrix}
N^t \\
\{D^{mt}\}_{m \in \mathcal{M}}
\end{bmatrix} u^t.
\end{align*}
\end{lem}



The optimization problem is defined as follows.

\begin{pb}
\label{Prob_1}
For the model described in section \ref{Team_mo}, given that information structure is not partially nested and the substitutability assumption (Assumption \ref{asum_subs_team}) holds, find the team strategy $\mathbf{g}= (g^1,\ldots, g^n)$ that minimizes the expected cost given by \eqref{cost_team}. 
\end{pb}
\vspace{-5mm}
\subsection{Partially nested expansion of the information structure}
In order to solve Problem \ref{Prob_1}, we will consider a partially nested expansion of the information structure.   This expansion is constructed by simply providing each team member $i$ with the information of members in $\mathcal{C}^i$.
We thus formulate the following problem. 

\begin{pb}
\label{Prob_2}
Solve Problem \ref{Prob_1} under the assumption that the information available to member $i$, $i \in \mathcal{M}$, is
\begin{align}
\label{pn_inf.}
\tilde{Z}^i = \begin{bmatrix}
Z^i \\
Z^{\mathcal{C}^i}
\end{bmatrix}.
\end{align}
\end{pb}

\begin{lem}
\label{partially_nestedness_P2}
The information structure of Problem \ref{Prob_2} is partially nested.
\end{lem}

\begin{proof}
It can be easily established that Problems 1 and 2 have the same precedence relationships. That is, $j$ is a precedent of $i$ in Problem 1 if and only if $j$ is a precedent of $i$ in Problem 2. Further, if $k$ is a precedent of $i$ (in both problems), then, by construction (see \eqref{pn_inf.}),  $\tilde{Z}^i$ contains $Z^k$.  

 Now, suppose that $j$ is a precedent of $i$ in Problem 2. To establish partial nestedness of the information structure in Problem 2, we need to establish that $\tilde{Z}^j = \VVEC(Z^j, Z^{\mathcal{C}^j}) \subset_v \tilde{Z}^i$. We already know that $Z^j \subset_v \tilde{Z}^i$. Further,  any $k \in \mathcal{C}^j$ is a precedent of $j$ and since $j$ is a precedent of $i$, it follows that $k$ is a precedent of $i$. Thus, $\tilde{Z}^i$ must contain $Z^k$.  This establishes that $\tilde{Z}^j  \subset_v \tilde{Z}^i$ and hence the information structure in Problem 2 is partially nested.
\end{proof}
\vspace{-3mm}
\begin{rem}
The idea of considering an expanded information structure and using strategies in the expansion to investigate optimal strategies in the original information structure has been used before \cite[Section 3.5.2]{YukselBasar:2013}, \cite{Yuksel:2009, chu1971team, mahajan_yuksel}. In some cases \cite[Section 3.5.2]{YukselBasar:2013}, \cite{Yuksel:2009}, it is shown that the expansion is redundant as far as strategy optimization is concerned: an optimal strategy is found in the expanded information structure that is  implementable in the original information structure. In \cite{chu1971team, mahajan_yuksel}, an optimal strategy found in the expanded structure may not be directly implementable in the original information structure but, under some conditions, it can be used to construct  an optimal strategy in the original information structure  (see Section \ref{sec:discussion}). As discussed below, our use of information structure expansion differs from both these approaches.
\end{rem}
\vspace{-5mm}
\subsection{Main Results}
Our main result relies on showing that we can construct an optimal team strategy in Problem \ref{Prob_1} from an optimal team strategy in the partially nested expansion of this problem (Problem \ref{Prob_2}). We start with the following assumption.

\begin{assum}
\label{linear_strategy_P2}
In Problem \ref{Prob_2} with partially nested information structure, there exists an  optimal  team  strategy $\boldsymbol\gamma_0 = (\gamma^1_0, \gamma^2_0, \ldots, \gamma^n_0)$  that is a linear function of the members' information and is given as
\begin{align}
\label{opt_st_p2}
{U}^i = \gamma^i_0(\tilde{Z}^i) = K^{ii}_0 Z^i + \sum_{j \in \mathcal{C}^i} K^{ij}_0 Z^j, ~~~ i \in \mathcal{M}.
\end{align}
\end{assum}


\vspace{-4mm}
\begin{rem}\label{rem:static}
The information available to each member of Problem \ref{Prob_2} can be written as,
\begin{align}
\label{partially_nested.inf.struct.}
\tilde{Z}^i = \tilde{H}^i \Xi + \sum_{j \in \mathcal{M} } \tilde{D}^{ij} U^j  
\hspace{9mm} \forall i \in \mathcal{M}
\end{align}
where $\tilde{H}^i = \{H^m\}_{m \in \{i\} \cup {\mathcal{C}^i}}$ and $\tilde{D}^{ij} = \{D^{mj}\}_{m \in \{i\} \cup {\mathcal{C}^i}}$.
Since Problem \ref{Prob_2} is a partially nested LQG problem, \cite[Theorem 1]{HoChu:1972} shows that it is equivalent to a static LQG team problem with the following information structure:

\vspace{-4mm}
\begin{align}
\label{partially_nested.inf.struct._2}
\hat{Z}^i = \tilde{H}^i \Xi  \hspace{9mm} \forall i \in \mathcal{M}.
\end{align}

\vspace{-1mm}
In particular, a linear strategy is optimal for Problem \ref{Prob_2} iff a linear strategy is optimal for the equivalent static team.
If the static LQG team has a cost function that is strictly convex in the team decision, then the optimal strategy of each team member is linear in the information of this member \cite{Radner:1962}. However, the cost function of \eqref{cost_team} is not strictly convex since $N^{\tp}N$ is not positive definite. Hence, the result of \cite{Radner:1962} cannot be applied here. 

According to \cite{AsghariNayyar_report}, for static LQG team problems with a cost function that is convex (not necessarily strictly convex) in the team decision, the linear team strategy $\gamma^i(\hat{Z}^i) = \Pi^i \hat{Z}^i$ for all $i \in \mathcal{M}$ is optimal if the following linear system of equations has a solution for $\Pi^i$, $i \in \mathcal{M}$,

\vspace{-5mm}
\begin{align}
\label{equation_system_linearity}
\sum_{j=1}^n  (N^i)^{\tp} N^j \Pi^j  \Sigma_{\hat{Z}^j \hat{Z}^i} = - (N^i)^{\tp} M \Sigma_{\Xi \hat{Z}^i} \hspace{6mm} \forall i \in \mathcal{M}
\end{align}
where $\Sigma_{\hat{Z}^j \hat{Z}^i} = \ee [\hat{Z}^j (\hat{Z}^i)^{\tp}]$ and $\Sigma_{\Xi \hat{Z}^i} = \ee [\Xi (\hat{Z}^i)^{\tp}] $. Therefore, Assumption \ref{linear_strategy_P2} is true as long as \eqref{equation_system_linearity} has a solution.
\end{rem}

\begin{rem}
Since team members in Problem \ref{Prob_2} have more information than the corresponding members in Problem \ref{Prob_1}, it follows that the optimal expected cost in Problem \ref{Prob_2} is a lower bound on the optimal expected cost in Problem \ref{Prob_1}.
\end{rem}

\vspace{-3mm}

\begin{theo}
\label{theorem_1}

\textcolor{black}{Under Assumptions 1 and 2,} there exist linear strategies in Problem \ref{Prob_1}, given as 
\vspace{-1mm}
\[U^i = \Gamma^iZ^i, ~~~i \in \mathcal{M},\]  that achieve the same expected cost as the optimal strategies in Problem \ref{Prob_2}. Consequently, the strategies  $U^i = \Gamma^iZ^i$, $i \in \mathcal{M}$, are optimal strategies in Problem \ref{Prob_1}.
\end{theo}

\vspace{-4mm}
\subsection{Proof of Theorem \ref{theorem_1}}
Let $\boldsymbol{\gamma}_0 := \{\gamma^1_0,\ldots, \gamma^n_0 \}$ be an optimal team strategy of Problem \ref{Prob_2} as given in \eqref{opt_st_p2}. These strategies \textcolor{black}{may} violate the information structure of Problem \ref{Prob_1}
since they use some information not available to members in Problem \ref{Prob_1}. 
For $i \in \mathcal{M}$ and $r \in \mathcal{C}^i$,  we say that $\gamma^i_0$ uses $Z^r$ if the matrix $K^{ir}_0 \neq 0$. We define $\mathcal{E}^i_0 = \lbrace r : r \in \mathcal{C}^i$ and $\gamma^i_0$ uses $Z^r \rbrace$. The cardinality of this set is referred to as the number of information structure violations of the strategy $\gamma^i_0$. Clearly,  $\left\vert{\mathcal{E}^i_0}\right\vert \leq \left\vert{\mathcal{C}^i}\right\vert$.

We will use the optimal team strategy  $\boldsymbol \gamma_0$ in Problem \ref{Prob_2} to construct an optimal team strategy in Problem \ref{Prob_1}. We will proceed iteratively. At each step of the iteration, we will construct an equivalent team strategy for which the total number of information structure violations across all members is one less than the previous team strategy.  This iterative process, described in \textcolor{black}{Algorithm 1}, can be summarized as follows: At the beginning of $l^{th}$ iteration, we are given a linear team strategy $\boldsymbol \gamma_l$ of Problem \ref{Prob_2}.  We consider members $t$ and $s$ such that $s \in \mathcal{C}^t$ and $\gamma_l^t$ uses $Z^s$. This represents an information structure violation for Problem \ref{Prob_1}. We carry out a strategy transformation, referred to as Procedure 1 and described in detail below, to obtain a new team strategy $\boldsymbol \gamma_{l+1}$ which has the same performance as $\boldsymbol\gamma_l$ but $\gamma^t_{l+1}$ does not use $Z^s$. Thus, the number of information structure violations is reduced by one. The $0^{th}$ iteration starts with  the optimal team strategy of Problem \ref{Prob_2} as given in \eqref{opt_st_p2}. The process terminates after $l^* = 
\sum_{j=1}^n \left\vert{\mathcal{E}^j_0}\right\vert
 \leq \sum_{j=1}^n \left\vert{\mathcal{C}^j}\right\vert$ iterations at which point the number of information structure violations has been reduced to $0$. 
 As a result, in the team strategy 
$\boldsymbol{\gamma}_{l^*}$, $\gamma_{l^*}^i$ only uses $Z^i$ (which is available to 
member $i$ in Problem 1). Furthermore, $\boldsymbol{\gamma}_{l^*}$ has the same performance as $\boldsymbol{\gamma}_{0}$. 
Therefore, the team strategy $\boldsymbol{\gamma}_{l^*}$ is optimal for Problem 1 and for each team member $i$, $\gamma_{l^*}^i$
 is linear in the information of member~$i$. 
\vspace{-3mm}
\begin{algorithm}
\begin{small}
\caption{}\label{euclid}
\begin{algorithmic}[1]
\State $l=0$;  (Iteration Number)
\State $\boldsymbol{\gamma}_0$ from \eqref{opt_st_p2};
\State $\mathcal{E}^i_0 = \lbrace r : r \in \mathcal{C}^i$ and $\gamma^i_0$ uses $Z^r \rbrace$,  $i=1,\ldots,n$;
\State \textbf{for} $t = 1$ to $n$
\State \hspace{0.2cm}\textbf{while} $\left\vert{\mathcal{E}^t_{l}}\right\vert \neq 0$
\State \hspace{0.3cm}$s = \min \lbrace r : r \in \mathcal{E}^t_l  \rbrace$; 
\State \hspace{0.3cm}Find new team strategy $\boldsymbol\gamma_{l+1}$ from $\boldsymbol\gamma_{l}$ according to \textbf{Procedure 1};
\State \hspace{0.3cm}$\mathcal{E}^t_{l+1} := \mathcal{E}^t_l \setminus \lbrace s \rbrace$;
\State \hspace{0.3cm}$\mathcal{E}^j_{l+1} := \mathcal{E}^j_l $, $j =1,\ldots, n, ~j \neq t$;
\State \hspace{0.3cm}$l= l+1$;
\State \hspace{0.2cm}\textbf{end while}
\State \textbf{end for}
\end{algorithmic}
\end{small}
\end{algorithm} 
 
 \vspace{-3mm}
We now describe Procedure 1 and then show that it preserves the expected cost. 
\textbf{Procedure 1:}
Given linear team strategy $\boldsymbol{\gamma}_l$ in Problem \ref{Prob_2} and team members $t$ and $s$ such that $s \in \mathcal{C}^t$ and $\gamma_l^t$ uses $Z^s$.  The set $\mathcal{E}^t_l = \lbrace r : r \in \mathcal{C}^t$ and $\gamma^t_l$ uses $Z^r \rbrace$ represents the information structure violations for member $t$ under $\gamma^t_l$. 

The strategy of member $t$ can be written as:
\begin{align}
\label{l_0_1}
\gamma^t_l (\tilde{Z}^t) = K^{tt}_l Z^t + \sum_{j \in \mathcal{E}^t_l \setminus \{s\} } K^{tj}_l Z^j + K^{ts}_l Z^s.
\end{align}

\vspace{-1mm}
According to Assumption \ref{asum_subs_team}, there exists a substituting member $k$ for the critical pair $(s,t)$. We construct new strategies for members $t$ and $k$ as follows:
\begin{align}
\label{l_0_2}
\gamma^t_{l+1}(\tilde{Z}^t) &= \gamma^t_{l}(\tilde{Z}^t) - K^{ts}_l Z^s = K^{tt}_l Z^t + \sum_{j \in \mathcal{E}^t_l \setminus \{s\} } K^{tj}_l Z^j, 
\end{align}
\vspace{-8mm}
\begin{align}\label{l_0_2_a}
\hspace{-30mm} \gamma^k_{l+1}(\tilde{Z}^k) &= \gamma^k_{l}(\tilde{Z}^k) + 
\Lambda^{kst}
K^{ts}_l Z^s
 \nonumber \\
&= K^{kk}_l Z^k + \sum_{j \in \mathcal{E}^k_l} K^{kj}_l Z^j +
\Lambda^{kst}
K^{ts}_l Z^s \nonumber \\
& = K^{kk}_{new} Z^k + \sum_{j \in \mathcal{E}^k_l} K^{kj}_l Z^j
\end{align}
where  $K^{kk}_{new}$ is derived from $K^{kk}_l$ and $\Lambda^{kst}K^{ts}_l$ because $Z^s \subset_v Z^k$. 

At the end of the procedure, the strategies of members can be written as follows,
\begin{itemize}
\item For member $t$, $\gamma^t_{l+1}(\tilde{Z}^t) = K^{tt}_{l+1} Z^t + \sum_{j \in \mathcal{E}^t_l \setminus \{s\} } K^{tj}_{l+1} Z^j$ where $K^{tt}_{l+1} = K^{tt}_l$ and $K^{tj}_{l+1} = K^{tj}_{l}$ for ${\small j \in \mathcal{E}^t_l \setminus \{s\}}$.
\vspace{5pt}
\item For member $k$, $\gamma^k_{l+1}(\tilde{Z}^k) = K^{kk}_{l+1} Z^k + \sum_{j \in \mathcal{E}^k_l} K^{kj}_{l+1} Z^j$ \\where $K^{kk}_{l+1} = K^{kk}_{new}$ and $K^{kj}_{l+1} = K^{kj}_{l}$ for $j \in \mathcal{E}^k_l$.
\vspace{2pt}
\item For all other members $r \in \mathcal{M} \setminus \{t,k\}$, \\ 
$\gamma^r_{l+1}(\tilde{Z}^r) = K^{rr}_{l+1} Z^r + \sum_{j \in \mathcal{E}^r_l} K^{rj}_{l+1} Z^j = \gamma^r_{l} (\tilde{Z}^r)$ 
\\ where $K^{rr}_{l+1} = K^{rr}_l$ and $K^{rj}_{l+1} = K^{rj}_{l}$ for $j \in \mathcal{E}^r_l$.
\end{itemize}
%

By construction, member $t$'s new strategy is no longer using $Z^s$ while every other member's new strategy is using the same information as before. Thus, the total number of information structure violations has been reduced by one. 
 ~\\
 
\vspace{-2mm}
To show that Procedure 1 preserves the expected cost, we need to show that the team strategy $\boldsymbol\gamma_{l+1}$ achieves the same expected cost as the team strategy  $\boldsymbol\gamma_{l}$. We start with the following claim.

\begin{claim}
\label{claim_u}
Let  us denote the team decision under team strategies $\boldsymbol\gamma_{l+1}$ and $\boldsymbol\gamma_{l}$ in Problem \ref{Prob_2} by $U\big|_{\boldsymbol\gamma_{l+1}}$ and $U\big|_{\boldsymbol\gamma_{l}}$ respectively. Then,
\begin{align}
\label{claim_performance_an}
NU\big|_{\boldsymbol\gamma_{l+1}} 
= NU\big|_{\boldsymbol\gamma_{l}} .
\end{align}
\end{claim}
\begin{proof}
See Appendix \ref{proof_claim}.
\end{proof}
\begin{rem}
Under the team strategies $\boldsymbol\gamma_{l+1}$ and $\boldsymbol\gamma_{l}$, $U\big|_{\boldsymbol\gamma_{l+1}}$ and $U\big|_{\boldsymbol\gamma_{l}}$ are linear functions of $\Xi$ and are, therefore, well-defined random vectors. The  equality in \eqref{claim_performance_an} should be interpreted as follows: for every realization $\xi$ of $\Xi$, the realizations of $NU\big|_{\boldsymbol\gamma_{l+1}} $ and $NU\big|_{\boldsymbol\gamma_{l}} $ are equal to each other.
\end{rem}

Based on Claim \ref{claim_u}, the following equality holds for every realization of the random vectors involved,
\begin{align}
\label{claim_condition}
M\Xi + NU\big|_{\boldsymbol\gamma_{l+1}} 
= M\Xi + NU\big|_{\boldsymbol\gamma_{l}}.
\end{align}
Consequently, the expected costs under the team strategies $\boldsymbol\gamma_{l+1}$ and $\boldsymbol\gamma_{l}$ are identical.
\vspace{-2mm}
\subsection{Discussion} \label{sec:discussion}
{\color{black}

The proof of Theorem \ref{theorem_1} shows that under Assumptions \ref{asum_subs_team} and \ref{linear_strategy_P2} an optimal team strategy in the partially nested Problem \ref{Prob_2} that violates the information structure of Problem \ref{Prob_1} can be \emph{transformed into an equivalent strategy that can be implemented in Problem \ref{Prob_1}}. The idea of utilizing optimal strategies in a partially nested expansion of a team problem to construct equivalent strategies in the original problem was used in \cite{chu1971team} as well. We paraphrase \cite[Theorem 2]{chu1971team} below:

\emph{\cite[Theorem 2]{chu1971team}:} 
Consider the setup of Problem \ref{Prob_1} but without Assumption \ref{asum_subs_team}.
 Let $\boldsymbol{\gamma_0} := (\gamma^1_0,\ldots, \gamma^n_0)$ be an optimal team strategy in its partially nested expansion.  We will assume that $\boldsymbol{\gamma_0}$ violates the information structure of Problem \ref{Prob_1}\footnote{This assumption is not made in \cite{chu1971team}. But it is clear that if $\boldsymbol{\gamma_0}$ does not violate the information structure of Problem \ref{Prob_1}, then it is an optimal strategy in that problem and no further construction is needed.}.  Under this team strategy, let $p^{i*}$ be the \emph{composite control function} from $\Xi$ to $U^i$ defined such that $p^{i*}(\Xi) = \gamma^i_0(Z^i)$.   Define functions $g^i: \mathbb{R}^{d_{\xi}} \mapsto \mathbb{R}^{d_{z}^i}$ for $i \in \mathcal{M}$ as
\begin{align} \label{eq:g_function}
 g^i(\xi) &:= \eta^i(\xi, p^{1*}(\xi),p^{2*}(\xi),\ldots,p^{i-1*}(\xi)) \notag
 \\
&:= H^i \xi + \sum_{j <i} D^{ij} p^{j*}(\xi).
\end{align}

\vspace{-3mm}
Suppose there exist functions  
$\boldsymbol{r} = (r^{1}, \ldots, r^{n})$ where $r^{i}: \mathbb{R}^{d_{z}^i} \mapsto \mathbb{R}^{d_{u}^i}$ for $i \in \mathcal{M}$ such that,
\begin{align}
p^{i*}(\Xi) = r^i (g^i(\Xi)), \quad \forall i \in \mathcal{M}. \label{eq:r_func}
\end{align}
Then,  \cite[Theorem 2]{chu1971team} states that $\boldsymbol{r}$ is an optimal team strategy for the original non-partially-nested problem.

In comparing our result to \cite[Theorem 2]{chu1971team}, the following key observations can be made:
\begin{enumerate}
\item The substitutability assumption required for our result is a condition placed on the information structure of Problem \ref{Prob_1} and on the parameters in the cost and observation equations (namely, the matrices $N^i,D^{ij}$, $i,j\in
\mathcal{M}$).  The condition required for the result in \cite{chu1971team}, on the other hand, is  a requirement that an optimal team strategy in the partially nested expansion must satisfy. Clearly, our result and the result in \cite{chu1971team} require conditions of very different nature.

\item Using an optimal strategy $\boldsymbol{\gamma_{0}}$ in the expanded structure, the result in \cite{chu1971team} constructs a team strategy $\boldsymbol{\gamma_{eq}}$ for the original team problem in a manner that ensures that for each $i \in \mathcal{M}$, $U^i\big|_{\boldsymbol\gamma_{0}} 
= U^i\big|_{\boldsymbol\gamma_{eq}}$.  In contrast, the strategies constructed in our proof ensure that $NU\big|_{\boldsymbol\gamma_{l+1}} 
= NU\big|_{\boldsymbol\gamma_{l}} $ (see Claim 1).  In other words, the transformation in \cite{chu1971team} ensures that each member's action $U^i$ is the same random variable under the original and transformed strategies. In contrast,  under our transformation,  a member's action may become a different random variable but the combined effect of the team members' actions on cost as captured by the term $NU$ remains unchanged.



\item We next consider the following question: Suppose we solve Problem \ref{Prob_2} under Assumptions 1 and 2 and find an optimal team strategy\footnote{There may be many optimal strategies. We pick one arbitrarily.} $\boldsymbol{\gamma_0}$ that violates the information structure of Problem \ref{Prob_1}. We then construct the composite control functions $p^{i*}, i \in \mathcal{M},$ under $\boldsymbol{\gamma_0}$.  Then, do there always exist functions $\boldsymbol{r} = (r^{1}, \ldots, r^{n})$ satisfying \eqref{eq:r_func}? In other words, are our assumptions sufficient conditions for $\boldsymbol{\gamma_0}$ to satisfy the conditions imposed in \cite[Theorem 2]{chu1971team}? The answer is no as the following  examples demonstrate:

\begin{examp}
\label{Comparison}
Consider Problem 1 where the control actions are one-dimensional, $\Xi \sim \mathcal{N}(0, 1)$, and
\vspace{-1mm}
\begin{align}
\mathcal{M}&= \{1,2,3\}, \quad N = \begin{bmatrix} 0 & 1 & 1  \end{bmatrix}, \quad M = 1 \notag \\
& \hspace{6mm} Z^1= \Xi, \quad Z^2 = U^1, \quad Z^3 = \Xi.
\label{comparison_ex_model}
\end{align}
It is straightforward to see that  $(1,2)$ is a critical pair and that member $3$ is a substituting member for this critical pair.  The information structure for the partially nested expansion of this example is 
\begin{align}
\label{inf_structure_PNX}
\tilde{Z}^1 = Z^1= \Xi, 
\tilde{Z}^2 = 
 \begin{bmatrix}
Z^2 \\ Z^1
\end{bmatrix} =
\begin{bmatrix}
U^1 \\ \Xi
\end{bmatrix},
 \tilde{Z}^3 = Z^3 = \Xi.
\end{align}
An optimal strategy in the partially nested expansion is 
\begin{align}
\label{eq_optimal_static_example1}
&U^1 =\gamma_0^1(\tilde{Z}^1) = 0,  \hspace{8mm} U^2 =\gamma_0^2(\tilde{Z}^2) = -0.5 Z^1, \notag \\
& \hspace{18mm} U^3 =\gamma_0^3(\tilde{Z}^3) = -0.5 Z^3.
\end{align}
The above strategy results in the lowest possible expected cost of $0$.
The composite control functions under the above strategy are 
\vspace{-2mm}
\begin{align}
p^{1*}(\Xi) = 0, \quad p^{2*}(\Xi) =  -0.5 \Xi, \quad p^{3*}(\Xi) =   -0.5 \Xi,
\end{align}
and the functions $g^i$ defined in \eqref{eq:g_function} are
\begin{align}
g^{1}(\Xi) = \Xi, \quad g^{2}(\Xi) =  0, \quad g^{3}(\Xi) = \Xi.
\end{align}
It is clear that there is no function $r^2:\mathbb{R} \mapsto \mathbb{R}$ such that 
\begin{align}
p^{2*}(\Xi) = r^2 (g^2(\Xi)). \label{eq:r_func2}
\end{align}
\end{examp}

In the above example, one can easily find other optimal strategies in the partially nested expansion that would satisfy the conditions of \cite[Theorem 2]{chu1971team}.  The point  we wish to make is that for an example that meets our assumptions,  there may be some  optimal strategies in the partially nested expansion that do not satisfy the conditions of \cite[Theorem 2]{chu1971team}. Our Algorithm 1, on the other hand, works with any linear optimal strategy in the partially nested expansion.

\begin{examp}\label{Comparison2}
Consider Problem 1 where the control actions are one-dimensional,  $ \Xi = \begin{bmatrix} \Xi^1 \\ \Xi^2 \end{bmatrix} \sim \mathcal{N}(0, \begin{bmatrix} 2 & -1 \\ -1 & 2  \end{bmatrix})$ and 
\begin{align}
\mathcal{M}&= \{1,2,3\}, \quad N = \begin{bmatrix} 2 & -1 & -1 \\ 1 & 2 & 2  \end{bmatrix}, 
\quad M =  \begin{bmatrix} 1 & 0 \\ 0 & 1  \end{bmatrix}
\notag \\
& \hspace{6mm} Z^1= \Xi^2, \quad Z^2 = U^1, \quad Z^3 = \Xi^2.
\label{comparison_ex_model_2}
\end{align}
It is straightforward to see that  $(1,2)$ is a critical pair and that member $3$ is a substituting member for this critical pair.  The information structure for the partially nested expansion of this example is 
\begin{align}
\label{inf_structure_PNX_2}
\tilde Z^1 = Z^1= \Xi^2, 
\tilde Z^2 = 
 \begin{bmatrix}
Z^2 \\ Z^1
\end{bmatrix} =
\begin{bmatrix}
U^1 \\\Xi^2
\end{bmatrix}, 
\tilde Z^3 = Z^3 = \Xi^2.
\end{align}
Since this information structure is partially nested,  it is equivalent to a static team with the following information structure,
\begin{align}
\label{inf_structure_S}
\hat Z^1 = \Xi^2, \quad 
\hat Z^2 =\Xi^2, \quad 
\hat Z^3 = \Xi^2.
\end{align}
According to Remark \ref{rem:static}, the linear team strategy $U^i = \gamma^i(\hat{Z}^i) = \Pi^i \hat{Z}^i = \Pi^i \Xi^2$ for $i \in \mathcal{M}$ is optimal if the following linear system of equations has a solution for $\Pi^i$, $i \in \mathcal{M}$,
\begin{align}
\label{equation_system_example}
10\Pi^1 = 0, \quad  10\Pi^2 + 10\Pi^3 = -5.
\end{align}
One solution provides the following strategies:
\begin{align}
\label{optimal_static_example}
U^1 = 0, \quad U^2 =  \Xi^2, \quad U^3 = -1.5 \Xi^2.
\end{align}
\eqref{optimal_static_example} can be written as follows under the information structure of \eqref{inf_structure_PNX_2},
\begin{align}
\label{eq_optimal_static_example}
&U^1 =\gamma_0^1(\tilde{Z}^1) = 0,  \hspace{8mm} U^2 =\gamma_0^2(\tilde{Z}^2) =  Z^1, \notag \\
& \hspace{12mm} U^3 =\gamma_0^3(\tilde{Z}^3) = -1.5 Z^3.
\end{align}
The team strategies of \eqref{eq_optimal_static_example} cannot be implemented in the original non-partially-nested information structure because $U^2$ uses $Z^1$ while $Z^1 \not \subset_v Z^2$.
We now follow the procedure of Algorithm 1 and use $\boldsymbol{\gamma}_0 = (\gamma_0^1,\gamma_0^2, \gamma_0^3)$ from \eqref{eq_optimal_static_example} to find optimal team strategies that can be implemented in the original information structure. Since there is only one information structure violation under $\boldsymbol{\gamma}_0$, we obtain the desired strategies after one iteration: 
\begin{align}
\label{1st_iteration_example}
&U^1 =\gamma_1^1(\tilde{Z}^1) = 0,  \hspace{8mm} U^2 =\gamma_1^2(\tilde{Z}^2) = 0, \notag \\
& \hspace{0mm} U^3 =\gamma_1^3(\tilde{Z}^3) =  -0.5 Z^3.
\end{align}

To compare our approach with that of \cite{chu1971team}, note that the composite control functions under $\boldsymbol{\gamma}_0$ are 
\begin{align}
p^{1*}(\Xi) = 0, \quad p^{2*}(\Xi) =   \Xi^2, \quad p^{3*}(\Xi) =   -1.5 \Xi^2.
\end{align}

and the functions $g^i$ defined in \eqref{eq:g_function} are
\begin{align}
g^{1}(\Xi) = \Xi^2, \quad g^{2}(\Xi) =  0, \quad g^{3}(\Xi) = \Xi^2.
\end{align}
Clearly, there is no function $r^2$ such that $p^{2*}(\Xi)  = r^2 (g^2(\Xi) )$. 
\end{examp}
\item Under our assumptions, the strategies $\Gamma^iZ^i$, $i \in \mathcal{M}$, of Theorem \ref{theorem_1} are optimal for both Problems 1 and 2. If these are used as $\boldsymbol{\gamma_{0}}$ in \cite[Theorem 2]{chu1971team}, then it can be shown that $r^i=\Gamma^i$ will satisfy  \eqref{eq:r_func}.  Of course, if we know $\Gamma^i$, $i \in \mathcal{M}$, already, then there is no need to carry out the transformation of \cite[Theorem 2]{chu1971team}.
\item  Finally,  \cite[Problem A in Section IV]{chu1971team} presents an example where the conditions of \cite[Theorem 2]{chu1971team} hold, but the substitutability assumption does not.  Thus, our assumptions do not provide necessary conditions for the conditions imposed in \cite[Theorem 2]{chu1971team}.

\end{enumerate}
The core idea of substitutability is also conceptually different from the conditional independence related properties of stochastic nestedness \cite{Yuksel:2009} and P-quasiclassical information structures \cite{mahajan_yuksel} that have been used for some non-partially-nested problems. In the non-partially-nested models of \cite{Yuksel:2009} and \cite{mahajan_yuksel}, one can identify an agent's ``missing information'' that prevents the information structure from being partially nested, i.e., if the agents knew their missing information, then the information structure would be partially nested.  These papers then rely on conditional independence like properties (of the relevant cost or state variables) to argue that given an agent's actual information, the missing information is irrelevant for making optimal decisions. We believe that this is very different from the essence of substitutability. Under our assumptions, it is not  the case that  the missing information  of agent $i$ is irrelevant for its decision. It's just that there is another agent present that knows the information missing at agent $i$ and can reproduce any effects on cost and observations that agent $i$ could have produced had it known its missing information. 
Let's reconsider Example \ref{Comparison} where the information structure is 
 $Z^1= \Xi, \quad Z^2 = U^1, \quad Z^3 = \Xi$. Suppose the following strategies are being used under this information structure:
\begin{align}
&U^1 =\gamma^1({Z}^1) = 0,  U^2 =\gamma^2({Z}^2) = 1, U^3 =\gamma^3({Z}^3)  =0.
\end{align}
Under the above strategies, the conditional expectation of the cost terms that involve $U^2,$ given $Z^2,U^2,$ can be computed to be
\begin{align}
\ee^{\boldsymbol{\gamma}}\left[
      U^2U^2 + 2U^2\Xi + 2U^2U^3 \mid Z^2,U^2\right] = 1.
\end{align}
On the other hand, if the same expectation is computed given $Z^1,U^1,Z^2,U^2$, we get
\begin{align}
\ee^{\boldsymbol{\gamma}}\left[
      U^2U^2 + 2U^2\Xi + 2U^2U^3 \mid Z^1,U^1,Z^2,U^2\right] = 1+2\Xi.
\end{align}
If the information structure was P-quasiclassical, then the two conditional expectations above should have been identical (see \cite[Definition 2]{mahajan_yuksel}). 
This demonstrates that Example \ref{Comparison} violates the definition of P-quasiclassical information structures  even though it satisfies our substitutability assumption.
}

\section{ Substituability in Decentralized LQG Control}\label{sec:OF}
\textcolor{black}{
Decentralized control problems in discrete time can be viewed as dynamic team problems by viewing a controller's actions at different time instants as the actions of distinct team members \cite{HoChu:1972}. Thus,  a decentralized control problem with $n$ controllers acting over a time horizon of duration $T$ can be seen as a dynamic team with $nT$ members, each member responsible for one control action. We will denote the team member corresponding to controller $i$'s action at time $t$ as \emph{member $i.t$}. We can then verify whether this dynamic team satisfies the assumptions of Section~\ref{dynamic_team} and if it does we can use an optimal team strategy in its partially nested expansion to find an optimal team strategy in the original team. The optimal strategy for member $i.t$ then naturally becomes the control strategy for controller $i$ at time $t$. Thus, non-partially-nested decentralized control problems whose dynamic team representations satisfy 
 Assumptions 1 and 2 can be solved using the analysis of Section~\ref{dynamic_team}. It is possible, however, to exploit the state structure in control problems to (a) simplify the verification of the substitutability assumption and (b) to find compact control strategies with recursively update-able sufficient statistics. We demonstrate this by considering the following problem.
 }

\subsection{System Model and Information Structure}\label{sec:lqgmodel}
We consider a decentralized control problem with $n$ controllers where
\begin{enumerate}
\item The state dynamics are given as 
\begin{equation}
\label{dynamics_eq}
X_{t+1} =AX_t + BU_t +W_t, \quad t=1,\ldots,T-1,
\end{equation}
where $X_t, W_t \in \mathbb{R}^{d_x}$, $U_t \in \mathbb{R}^{d_u}$ and $U_t = \VVEC(U^1_t,\ldots,U^n_t)$.

\item Each controller makes a noisy observation of the system state given as
\begin{align}
\label{equation:x28}
Y_{t}^i = C^{i}X_t + V_{t}^i, \hspace{10mm} i=1,\ldots,n.
\end{align}
 Combining (\ref{equation:x28}) for all controllers gives:
\begin{align}
\label{equation:x281}
Y_{t} = C
X_t + V_{t},
\end{align}
where $Y_t$ denotes $\textsl{vec}(Y_{t} ^1, Y_{t} ^2, \ldots, Y_{t} ^n)$ and $V_t$ denotes $\textsl{vec}(V_{t} ^1, V_{t} ^2, \ldots, V_{t} ^n)$ and $C$ is a matrix composed of $C^1,\ldots,C^n$ as row blocks.
\end{enumerate}
 The initial state $X_1$ and the noise variables $W_t, t=1,\ldots,T-1,$ and $V_t, t=1,\ldots,T-1,$ are mutually independent and jointly Gaussian with the following probability distributions:
\begin{align*}
X_1 \sim \mathcal{N}(0,\Sigma_x), \hspace{5mm} W_t \sim \mathcal{N}(0,\Sigma_w), \hspace{5mm} V_t \sim \mathcal{N}(0,\Sigma_v).
\end{align*}
The information available to the $i^{th}$ controller at time $t$ is:
\begin{align}
\label{equation:x31}
I_t ^i &= \lbrace Y_{1:t} ^i, U_{1:t-1} ^i \rbrace \hspace{10mm} i=1,\ldots,n. 
\end{align}
Each controller $i$, chooses  its action $U_t ^i $ according to $U_t^i = g_t^i(I_t^i)$. 
The collection $g^i = (g_1^i , \ldots, g_T^i)$ is called the control strategy of controller $i$. 
The performance of the control strategies of all controllers, $\boldsymbol{g}= (g^1, \ldots, g^n)$, is measured by the total expected cost over a finite time horizon:
\begin{align}
\label{equation:xx4}
\mathcal{J}(\boldsymbol{g})= \ee^{\boldsymbol{g}}\left[\sum_{t=1}^{T} 
       (MX_t + NU_t)^{\tp}(MX_t + NU_t)\right].
\end{align}

The optimization problem is defined as follows.
\begin{pb}
\label{Prob_3}For the model described above, find control strategy $\boldsymbol{g}= (g^1, \ldots, g^n)$ that minimizes the expected cost given by (\ref{equation:xx4}). 
\end{pb}

\subsection{Substitutability Assumption}
We  make the following assumption about the system.
\begin{assum}\label{asm:2}
For every vector $u =\VVEC(u^1,u^2,\ldots,u^n)$, there exist control actions $v^i = l^i(u)$ for controller $i$, $i=1,\ldots,n$, such that
\begin{equation}
\label{ass_1}
Bu = B\begin{bmatrix} 0 \\ \vdots \\v^i \\ \vdots \\ 0 \end{bmatrix} ~~\mbox{and}~~Nu = N\begin{bmatrix} 0 \\ \vdots \\v^i \\ \vdots \\  0 \end{bmatrix}.
\end{equation}
\end{assum}
We can write the $B$ and $N$ matrices in terms of their blocks as 
\[ B = \begin{bmatrix} B^1 &\ldots &B^n\end{bmatrix}, \quad  N = \begin{bmatrix} N^1 &\ldots  &N^n\end{bmatrix}. \]
An example of a system satisfying Assumption \ref{asm:2} is a two-controller LQG problem where the dynamics and the cost are functions only of the sum of the control actions, that is, ($u_t^1 + u_t^2$). This happens if $B^1 = B^2$ and $N^1= N^2$. In this case, using $v_t^1 = v_t^2 = u_t^1 + u_t^2$ satisfies \eqref{ass_1}.

{\color{black}
\begin{rem} More generally, Assumption \ref{asm:2} is satisfied iff the column spaces of  matrices 
$\begin{bmatrix}
B^i \\
N^i
\end{bmatrix}$, 
$i=1,\ldots,n,$ are identical.
\end{rem}

\begin{rem}
The substitutability assumption above (Assumption \ref{asm:2}) is really just a  compact representation of the substitutability assumption  of Section \ref{dynamic_team} (Assumption \ref{asum_subs_team}) with a specified substituting member for each critical pair.  To see this, first note that in the dynamic team representation of the control problem members $i.s$ and $j.t$ form a critical pair when $j\neq i$ and $s<t$. Secondly, member $i.t$ knows all the information of member $i.s$. To show that member~$i.t$ is a substituting member for the critical pair $(i.s, j.t)$, we just need to argue that for any action $u^j_{t}$, we can find an action $u^i_{t}$ that produces the same effect on total cost and future observations. Since the effect of  $u^j_{t}$ on the   cost at time $t$  is only through the term  $N^ju^j_{t}$ and its effects on the future costs and observations are only through $B^ju^j_{t}$, it suffices to ensure that  for any $u^j_{t}$, there exists $u^i_{t}$ such that
 \[ N^ju^j_{t} =  N^iu^i_{t} \mbox{~and~} B^ju^j_{t} =  B^iu^i_{t}.\]
 Combining the above for all $j \neq i$ gives the substitutability conditions of Assumption \ref{asm:2}.
\end{rem}
}

The following lemma is immediate from the theory of pseudo-inverses \cite{ben2006generalized}.
\begin{lem}
\label{lem_3}
If a solution $v^i$ to \eqref{ass_1} exists, it can be written as $v^i = \Lambda^i u$,  where
\vspace{-3mm}
\begin{align}
\label{Lamba_eq}
\Lambda^i = \begin{bmatrix}
B^i \\
N^i
\end{bmatrix}^{\dagger}
\begin{bmatrix}
B \\
N
\end{bmatrix}.
\end{align}
\end{lem}
\vspace{-4mm}
\subsection{A Centralized Problem}
In order to solve Problem \ref{Prob_3}, we  \textcolor{black}{would like to consider a  partially nested expansion of its information structure. Because members $i.s$ and $j.t$ form a critical pair  in the dynamic team representation of the control problem when $j\neq i$ and $s<t$, a partially nested expansion must give controller $j$ at time $t$ all the information of controller $i$ at any time $s<t$.   A convenient expansion that meets this requirement is the information structure of the  centralized   problem described below.}
\begin{pb}
\label{Prob_4}
For the model described above, assume that the information available to each controller is
\begin{align}
\label{equation:xx5}
\tilde{I}_t = \lbrace Y_{1:t}, U_{1:t-1} \rbrace.
\end{align}
Controller $i$ chooses its action according to strategy $U^i_t = g^i_t(\tilde{I}_t)$. The objective is to select control strategies that minimize (\ref{equation:xx4}).
\end{pb}
The following lemma follows directly from the problem descriptions above and well-known results for the centralized LQG problem with output feedback \cite{kumar_varaiya}. 
\begin{lem} \label{lem_5}
\begin{enumerate}
\item The optimal cost in Problem \ref{Prob_4} (with centralized information structure) is a lower bound on the optimal cost in Problem \ref{Prob_3} (with decentralized information structure).
\item The optimal strategies in Problem \ref{Prob_4} have the form of $U_t = K_t Z_t$ where $Z_t = \ee [X_t \vert \tilde{I}_t]$. $Z_t$ evolves according to the following equations:
\vspace{-2mm}
\begin{align}
\label{equation:x37}
Z_1 &= L_1 Y_1 \nonumber \\
Z_{t+1} &= (I - L_{t+1}C)(AZ_t + BU_t) + L_{t+1}Y_{t+1}. 
\end{align}
The matrices $L_t, t=1,\ldots,T$ can be computed apriori from the problem parameters. 


\end{enumerate}
\end{lem}

\vspace{-2mm}
\subsection{Main results}\label{sec:OF:mr}
In this section, we show that it is possible to construct optimal strategies in Problem \ref{Prob_3} from the optimal control strategy of Problem \ref{Prob_4}. 

\begin{theo}
\label{theorem_2}
Consider Problems \ref{Prob_3} and \ref{Prob_4}, and consider the optimal strategy, $U_t = K_t Z_t$,  of Problem \ref{Prob_4}. We write $L_{t+1}$ of Lemma \ref{lem_5} as $L_{t+1} = \begin{bmatrix} L_{t+1}^1 &L_{t+1}^2 &\ldots &L_{t+1}^n \end{bmatrix}$. The optimal control strategies of Problem \ref{Prob_3} can be written as
\begin{align}
\label{equation:x46}
U_t ^i =  \Lambda^i K_t S_t^i 
\end{align}
where $\Lambda^i$ is given by \eqref{Lamba_eq} and $S_t ^i$ satisfies the following update equations:
\vspace{-2mm}
\begin{align}
\label{equation:x41}
       S_{1} ^i &= L_1^i  Y_{1}^i        \nonumber \\
       S_{t+1}^i &= (I- L_{t+1}C)(A S^i_t+ B^iU^i_t)  + L_{t+1}^i Y_{t+1}^i.
\end{align}
Moreover, the optimal strategies in Problem \ref{Prob_3} achieve the same cost as the optimal strategies in Problem \ref{Prob_4}.
\end{theo}
~\\

\vspace{-8mm}
Observe that the strategies given by (\ref{equation:x46}) and \eqref{equation:x41} are valid control strategies under the information structure of Problem \ref{Prob_3} because they  depend only on $Y_{1:t}^i, U^i_{1:t-1}$ which are included in $I_t^i$. The states $S^i_t$ defined in  \eqref{equation:x41} are related to the centralized estimate $Z_t$ by the following result.
\begin{lem} \label{lem_6}
The centralized state estimate $Z_t$    and the states $S^i_t$ defined in \eqref{equation:x41} satisfy the following equation:
\begin{align}
\label{equation:x47}
Z_t =  \sum_{i=1}^{n}S_t ^i.
\end{align}
\end{lem}

\begin{proof}
We prove the result by induction. For $t=1$, from (\ref{equation:x37}), we have $Z_1 = L_1 Y_1$ and according to (\ref{equation:x41}), 
\begin{align}
\sum_{i=1}^{n}S_1 ^i = L_1^1  Y_{1}^1 + L_1^2  Y_{1}^2 + ... + L_1^n  Y_{1}^n = L_1 Y_1.
\end{align}
Now assume that $Z_t =  \sum_{i=1}^{n}S_t ^i $. We need to show that $Z_{t+1} =  \sum_{i=1}^{n}S_{t+1} ^i $. From (\ref{equation:x37}), it follows that 
\begin{equation}
Z_{t+1} = (I - L_{t+1}C)(AZ_t + BU_t) + L_{t+1}Y_{t+1}.
\end{equation}
 From (\ref{equation:x41}), we have
\begin{align}
\label{equation:x474}
&\sum_{i=1}^{n}S_{t+1} ^i = \sum_{i=1}^{n} [(I- L_{t+1}C)(A S^i_t+ B^i U^i_t)  + L_{t+1}^i Y_{t+1}^i]
\nonumber \\ &=
 (I- L_{t+1}C)(A  \sum_{i=1}^{n}S_t^i + \sum_{i=1}^n B^i U^i_t ) +\sum_{i=1}^{n} L_{t+1}^i Y_{t+1}^i] 
 \nonumber \\ &=
 (I - L_{t+1}C)(AZ_t  + BU_t)+ L_{t+1}Y_{t+1}.
\end{align}
Therefore, $Z_{t+1} =  \sum_{i=1}^{n}S_{t+1} ^i $.
\end{proof}

\begin{rem}
If $X_t = \VVEC(X^1_t,\ldots,X^n_t)$ and for each $i$ $Y^i_t=X^i_t$, it can be easily shown that $S^i_t = \textsl{vec}(0,\ldots,X^i_t,\ldots,0)$.
\end{rem}

The following result is an immediate consequence of Theorem \ref{theorem_2}.

\begin{cor}\label{cor:2}
For the model described in section \ref{sec:lqgmodel}, consider any information structure under which  the  information of controller $i$ at time $t$, $\hat{I}^i_t$, satisfies 
\[ \lbrace Y^i_{1:t}, U^i_{1:t-1} \rbrace \subseteq \hat{I}^i_t \subseteq \lbrace Y_{1:t}, U_{1:t-1} \rbrace,\]
for all $i=1,\ldots,n$ and $t=1,\ldots,T$.
Then, the optimal strategies in this information structure are the same as in Theorem \ref{theorem_2}.
\end{cor}

\vspace{-4mm}
\subsection{Proof of Theorem \ref{theorem_2}}\label{sec:OF:pot}
For notational convenience, we will describe the proof for $n=2$. If $U_t = K_t Z_t$ is the optimal control strategy of Problem \ref{Prob_4}, then from Lemma \ref{lem_6}, we have:
\begin{align}
\label{equation:x475}
U_t = K_t Z_t = K_t (S_t ^1 + S_t ^2).
\end{align}
We claim that the decentralized control strategies defined in Theorem \ref{theorem_2}, that is 
\begin{equation}
U_t =  \begin{bmatrix}
           U^1_t\\[0.3em]
           U^2_t
          \end{bmatrix} =\begin{bmatrix}
      \Lambda^1 K_t S_t^1\\[0.3em]
      \Lambda^2 K_t S_t^2        \end{bmatrix},  ~~~t=1,\ldots,T,\label{equation:x475a}
\end{equation}
 yield the same expected cost as the optimal centralized control strategies $U_t = K_t Z_t, t=1,\ldots,T$.

We first consider the control system under the centralized strategies.
We proceed sequentially to establish the claim by successively changing the control strategies at each time step.
Under the control strategies $U_t = K_t Z_t, t=1,\ldots,T$, the controlled system can be viewed as a linear system with $\VVEC(X_t,Z_t)$ as the state. 
We first change the control strategy at time $t=1$ from $U_1 = K_1 Z_1$ to the one given by \eqref{equation:x475a} and show that it doesn't change the instantaneous cost or the future evolution of the linear system.

Under control action $U_1 = K_1 Z_1$, we have $NU_1 = NK_1 Z_1$.
 Under control actions $U_1^1 = \Lambda^1 K_1 S^1_1, U_1^2 = \Lambda^1 K_1 S^2_1$, we have 
\begin{align}
\label{equation:x521}
NU_1 &= \begin{bmatrix} N^1 &N^2  \end{bmatrix} U_1 = N^1 U_1^1 + N^2 U_1^2 \nonumber \\
&= N^1 \Lambda^1 K_1 S^1_1 + N^2 \Lambda^2 K_1 S^2_1. 
\end{align}
From the substitutability assumption (Assumption \ref{asm:2}) and Lemma \ref{lem_3}, for any vector $u$,  $Nu  = N^i \Lambda^i u$. Therefore, 
\begin{align}
&N^1 \Lambda^1 K_1 S^1_1 = NK_1 S^1_1, \quad
N^2 \Lambda^2 K_1 S^2_1 = NK_1 S^2_1. 
\end{align}
  (\ref{equation:x521}) can now be written as,
\begin{align}
\label{equation:x522}
\hspace{-2mm}
N^1 \Lambda^1 K_1 S^1_1 + N^2 \Lambda^2 K_1S^2_1 = N(K_1 S^1_1 + K_1 S^2_1) = NK_1 Z_1,
\end{align}
where the last equality is true because $Z_1=S^1_1 +S^2_1$. Thus, the change in strategies at time $t=1$ does not affect the cost at time $t=1$.

The change in strategies at time $t=1$ affects the next state $\VVEC(X_2,Z_2)$ only through the term $BU_1$.
From the substitutability assumption (Assumption \ref{asm:2}) and  Lemma \ref{lem_3}, for any vector $u$,  $Bu = B^i \Lambda^i u$. Therefore, 
\begin{align}
\hspace{-2mm}
B^1 \Lambda^1 K_1 S^1_1 + B^2 \Lambda^2 K_1 S^2_1 = B(K_1 S^1_1 + K_1 S^2_1) = BK_1 Z_1.
\end{align}

The future state evolution is unaffected by the change in strategies at time $t=1$. Therefore, changing strategies at time $t=1$ from the centralized strategy to the one given  by \eqref{equation:x475a} does not change the expected cost. Proceeding in the same manner for all successive time instants establishes the claim.

\vspace{-1mm}
\section{Concluding Remarks}
We considered two problems, an LQG dynamic team problem and a decentralized LQG control problem and defined a property called substitutability in these problems.
For the non-partially-nested LQG dynamic team problem, we showed that under certain conditions an optimal strategy of each team member is linear in its information. For the non-partially-nested decentralized control problem under the substitutability assumption,  we showed that linear strategies are optimal and  we provided a complete state-space characterization of optimal strategies. 
Our results suggest that  substitutability can work as a counterpart of  the information structure requirements that enable simplification of dynamic teams and decentralized control problems. 

\vspace{-3mm}
\appendices

\section{Proof of Claim 1}
\label{proof_claim}
We want to show that the term $NU$ is the same under the team  strategies $\boldsymbol\gamma_{l+1}$ and $\boldsymbol\gamma_{l}$ of Problem \ref{Prob_2}. 
Under different team strategies, the information available to team members changes. Hence, we first need to show that $\tilde{Z}^i, \forall i \in \mathcal{M},$ is the same under $\boldsymbol\gamma_{l+1}$ and $\boldsymbol\gamma_{l}$. If we denote the information available to member $i$ in Problem 2 under team strategies $\boldsymbol\gamma_{l+1}$ and $\boldsymbol\gamma_{l}$ by 
$\tilde{Z}^i \big|_{\boldsymbol\gamma_{l+1}}$ and
$\tilde{Z}^i \big|_{\boldsymbol\gamma_{l}}$ 
respectively, we want to show that,
\begin{align}
\label{claim_performance}
\tilde{Z}^i \big|_{\boldsymbol\gamma_{l+1}}
= \tilde{Z}^i \big|_{\boldsymbol\gamma_{l}}   \hspace{5mm} \forall i \in \mathcal{M}.
\end{align}
According to \eqref{pn_inf.}, $\tilde{Z}^i $ is obtained from $\{Z^r,$  $r\leq i\}$.  Therefore, to show that \eqref{claim_performance} holds, it suffices to show that
\vspace{-0mm}
 \begin{align}
\label{claim_performance_equal}
Z^r \big|_{\boldsymbol\gamma_{l+1}}
= Z^r \big|_{\boldsymbol\gamma_{l}}   \hspace{5mm} \forall   r \in \mathcal{M}.
\end{align}
%

\vspace{-1mm}
According to Procedure 1, $\gamma^j_{l+1}$ is the same as $\gamma^j_{l}$  for $j \in \mathcal{M} \setminus \{t,k\}$. 
We, therefore, categorize team members into two groups:
\begin{itemize}
\item Group 1: $ \lbrace r\in \mathcal{M} : r \leq \min\{t, k\}    \rbrace$
\item Group 2: $ \lbrace r\in \mathcal{M} : r > \min\{t, k\}  \rbrace$
\end{itemize}

For $r$ in Group 1, $Z^r$ does not depend on the strategies of members $t$ and $k$. 
Therefore, for $r$ in Group 1, \eqref{claim_performance_equal}  holds.



We will show inductively  that for all $r \leq h$, $Z^r$ is the same under $\boldsymbol\gamma_{l}$ and $\boldsymbol\gamma_{l+1}$. For $h = \min \{t,k\}$, the statement holds  because we have shown that it holds for $r$ in Group 1. 


Now, assume that for some $\alpha \geq \min\{t,k\}$,    \eqref{claim_performance_equal} holds  for all $r \leq \alpha$ (induction hypothesis). This implies that \eqref{claim_performance} also holds for $r \leq \alpha$.  We now need to show that \eqref{claim_performance_equal} holds for $r=\alpha +1$. 

Suppose $\alpha+1 > t, \alpha+1 >k$. Under team strategy $\boldsymbol\gamma_{l+1}$, $Z^{\alpha +1}$ can be written as
\begin{align}
\label{inf_proof_1}
&Z^{\alpha +1} \big|_{\boldsymbol\gamma_{l+1}} = H^{\alpha +1} \Xi + \sum_{j <\alpha +1} D^{\alpha +1,j} U^j \big|_{\gamma^{j}_{l+1}}  \nonumber \\
& = H^{\alpha +1} \Xi + \sum_{j <\alpha +1 , j\neq t, j \neq k} D^{\alpha +1,j} \gamma^{j}_{l+1} (\tilde{Z}^j\big|_{\boldsymbol\gamma_{l+1}}) + \nonumber\\
&~~~ D^{\alpha +1,t} \gamma^{t}_{l+1} (\tilde{Z}^t\big|_{\boldsymbol\gamma_{l+1}}) + D^{\alpha +1,k} \gamma^{k}_{l+1} (\tilde{Z}^k\big|_{\boldsymbol\gamma_{l+1}}).
\end{align}

\vspace{-2mm}
 $\tilde{Z}^j, \tilde{Z}^t$ and $\tilde{Z}^k$  present in the right hand side of  \eqref{inf_proof_1} are the same under control strategies $\boldsymbol\gamma_{l}$ and $\boldsymbol\gamma_{l+1}$ by the induction hypothesis. 
Further, for $j \neq t, k,$ $\gamma^j_{l+1} = \gamma^j_l$.
Using these observations and \eqref{l_0_2} and \eqref{l_0_2_a}, \eqref{inf_proof_1} can be written as,
\begin{align}
\label{inf_proof_2}
&Z^{\alpha+1} \big|_{\boldsymbol\gamma_{l+1}} = H^{\alpha+1} \Xi  + \sum_{j <\alpha+1 , j\neq t, j \neq k} D^{\alpha+1,j} \gamma^{j}_{l} (\tilde{Z}^j\big|_{\boldsymbol\gamma_{l}}) + \nonumber\\
&
+ 
 D^{\alpha+1,t} \big(\gamma^t_{l}(\tilde{Z}^t\big|_{\boldsymbol\gamma_{l}}) - K^{ts}_l Z^s\big|_{\boldsymbol\gamma_{l}} \big)  \nonumber\\
&+ D^{\alpha+1,k} \big(\gamma^k_{l}(\tilde{Z}^k\big|_{\boldsymbol\gamma_{l}}) +\Lambda^{kst}
K^{ts}_l Z^s\big|_{\boldsymbol\gamma_{l}} \big).
\end{align}
According to Lemma \ref{lem_team} and the substitutability assumption, 
\begin{align}
\label{inf_proof_3}
&D^{\alpha+1,k} \Lambda^{kst}
K^{ts}_l Z^s\big|_{\boldsymbol\gamma_{l}} = D^{\alpha+1,t} K^{ts}_l Z^s\big|_{\boldsymbol\gamma_{l}}.
\end{align}
Using \eqref{inf_proof_3}, \eqref{inf_proof_2} can be simplified as,
\begin{align}
\label{inf_proof_4}
&Z^{\alpha+1} \big|_{\boldsymbol\gamma_{l+1}} = H^{\alpha+1} \Xi  + \sum_{j <\alpha+1 , j\neq t, j \neq k} D^{\alpha+1,j} \gamma^{j}_{l} (\tilde{Z}^j\big|_{\boldsymbol\gamma_{l}}) 
\nonumber \\
&+ 
 D^{\alpha+1,t} \gamma^{t}_{l} (\tilde{Z}^t\big|_{\boldsymbol\gamma_{l}}) + D^{\alpha+1,k} \gamma^{k}_{l} (\tilde{Z}^k\big|_{\boldsymbol\gamma_{l}})  
 = Z^{\alpha+1} \big|_{\boldsymbol\gamma_{l}}.
\end{align}

Thus, $Z^{\alpha+1} \big|_{\boldsymbol\gamma_{l+1}} = Z^{\alpha+1} \big|_{\boldsymbol\gamma_{l}}$ if $\alpha+1 >t, \alpha+1 >k$.
 If $t < \alpha+1 \leq k$ (alternatively $k < \alpha+1 \leq t$), we can employ arguments similar to above along with the fact that $D^{\alpha +1,k} =0$ (alternatively $D^{\alpha +1,t} =0$) to show \eqref{claim_performance_equal}  for $r = \alpha +1$.

Hence,  by induction, \eqref{claim_performance_equal} holds for all $r$ from $1$ to $n$.  Therefore, $Z^r$ and consequently $\tilde{Z}^r$ for $r \in \mathcal{M}$ are the same under team strategies $\boldsymbol\gamma_{l}$ and $\boldsymbol\gamma_{l+1}$. 

Now, we show that $NU$ is the same under team strategies $\boldsymbol\gamma_{l}$ and $\boldsymbol\gamma_{l+1}$. Under  $\boldsymbol\gamma_{l+1}$, $NU$ can be written as follows,
\begin{align}
\label{inf_proof_5}
&N U\big|_{\boldsymbol\gamma_{l+1}} = \sum_{j=1}^n N^j \gamma^{j}_{l+1} (\tilde{Z}^j)
\nonumber \\
&
=
\sum_{j \in \mathcal{M} \setminus \{t,k\}} N^j \gamma^{j}_{l+1} (\tilde{Z}^j)+ N^t \gamma^{t}_{l+1}(\tilde{Z}^t) + N^k \gamma^{k}_{l+1} (\tilde{Z}^k) \nonumber 
\\
&= \sum_{j \in \mathcal{M} \setminus \{t,k\}} N^j \gamma^{j}_{l}(\tilde{Z}^j) + N^t \big(\gamma^t_{l}(\tilde{Z}^t) - K^{ts}_l Z^s \big)  \nonumber \\
& 
+ N^k \big(\gamma^k_{l}(\tilde{Z}^k) + 
\Lambda^{kst}
K^{ts}_l Z^s
 \big)
\nonumber \\
&
=
\sum_{j \in \mathcal{M} \setminus \{t,k\}} N^j \gamma^{j}_{l} (\tilde{Z}^j)+ N^t \gamma^{t}_{l}(\tilde{Z}^t) + N^k \gamma^{k}_{l} (\tilde{Z}^k) 
\nonumber \\
&= N U\big|_{\boldsymbol\gamma_{l}}
\end{align}
where the penultimate equality is true because  Lemma \ref{lem_team} and the substitutability assumption provide that
\begin{align}
\label{inf_proof_8}
&N^{k} \Lambda^{kst}
K^{ts}_l Z^s = N^{t} K^{ts}_l Z^s.
\end{align}

\bibliographystyle{IEEEtran}
\bibliography{cdc15,collection}

\end{document}